\documentclass[11pt]{article}

\usepackage[T1]{fontenc}
\usepackage[english]{babel}

\usepackage[a4paper,margin=30mm]{geometry}
\usepackage{graphicx}
\usepackage{amssymb,amsmath,amsthm}
\usepackage{delarray}
\usepackage{enumitem}
\usepackage{tikz}
\usepackage{caption}
\usepackage{hyperref}
\usepackage{authblk}
\usepackage{stmaryrd}
\usepackage{ifthen}
\usepackage{calc}
\usepackage{float}
\usepackage[ruled,vlined,linesnumbered]{algorithm2e}
\usepackage{etoolbox}
\graphicspath{{figures/}}

\newtheorem{theorem}{\textbf{Theorem}}

\newtheorem{remark*}{Remark}

\newcommand{\PSPACE}{{\mathsf{PSPACE}}}

\newcounter{finaln}
\newcommand{\entiers}[2][2]{\ifthenelse{\equal{#1}{2}}
{\llbracket #2\rrbracket}
{\ifthenelse{\equal{#1}{0}}{
\ifthenelse{\equal{11}{\the\catcode`#2}}
{\{0,\ldots,#2 -1\}}
{\setcounter{finaln}{#2 -1}
\{0,\ldots,\thefinaln \}}}
{\{1,\ldots,#2\}}}}

\newcommand{\bool}{\{0,1\}}
\renewcommand{\int}[1]{[#1]}
\newcommand{\inter}[1]{G_{#1}}

\newcommand{\mmjblock}[2]{{#1}_{(#2)}}
\newcommand{\mmjoblock}[2]{#1_{\{#2\}}}

\newcommand{\mmj}[2]{{#1}_{#2}}
\newcommand{\mmjsub}[3]{\mmjoblock{#1}{#2}^{(#3)}}

\newcommand{\BPn}{\mathsf{BP}_n}

\newcommand{\BSn}{\mathsf{BS}_n}

\newcommand{\cyclen}{C^{n+}}
\newcommand{\bancyclen}{f^{n+}}

\newcommand{\arbitfn}{\smash{\hat{h}}^{n}}
\newcommand{\arbitf}{\smash{\hat{h}}}
\newcommand{\arbitmun}{\smash{\hat{\mu}}^{n}}
\newcommand{\arbitmu}{\smash{\hat{\mu}}}

\newcommand{\lcm}{\text{lcm}}

\title{Creation of fixed points\\
in block-parallel Boolean automata networks}
\author[1,2]{K{\'e}vin Perrot}
\author[1,2]{Sylvain Sen{\'e}}
\author[1]{L{\'e}ah Tapin}
\affil[1]{Aix Marseille Univ, CNRS, LIS, Marseille, France}
\affil[2]{Universit{\'e} publique, Marseille, France}

\date{}

\begin{document}

\renewcommand{\labelitemi}{$\bullet$}
\renewcommand{\labelitemii}{$\bullet$}
\setlist[itemize,enumerate]{nosep}

\maketitle
\begin{abstract}
  In the context of discrete dynamical systems and their applications,
  fixed points often have a clear interpretation.
  This is indeed a central topic of gene regulatory mechanisms modeled
  by Boolean automata networks (BANs),
  where a collection of Boolean entities (the automata)
  update their state depending on the states of others.
  Fixed points represent phenotypes such as differentiated cell types.
  The interaction graph of a BAN captures the architecture of dependencies among its automata.
  A first seminal result is that cycles of interactions (so called feedbacks)
  are the engines of dynamical complexity.
  A second seminal result is that fixed points are invariant under block-sequential update schedules,
  which update the automata following an ordered partition of the set of automata.
  In this article we study the ability of block-parallel update schedules (dual to the latter)
  to break this fixed point invariance property,
  with a focus on the simplest feedback mechanism: the canonical positive cycle.
  We quantify numerically the creation of new fixed points,
  and provide families of block-parallel update schedules generating exponentially many
  fixed points on this elementary structure of interaction.
\end{abstract}

\section{Introduction}

A Boolean automata network (BAN) is a finite discrete dynamical system
defined by $n$ local functions $f_i:\bool^n\to\bool$ for $i\in\{0,\dots,n-1\}$.
A configuration is an element of $\bool^n$, and an update schedule describes in which order
its components (the automata) are updated to compute the next configuration.
For example the parallel update schedule corresponds to the fully synchronous application
of its local function to each automaton of the network.
The structure of interactions is captured by a directed graph $\inter{f}$
on the vertices $\{0,\dots,n-1\}$, where the arcs represent the effective dependencies
among the automata.
When $\inter{f}$ is acyclic the dynamics converges to a unique fixed point~\cite{J-Robert1980},
therefore cycles of interactions are the engines of dynamical asymptotic heterogeneities.

BANs have many applications, in particular they are a classical model of gene regulatory networks
since the seminal papers of Kauffman~\cite{J-Kauffman1969} and Thomas~\cite{J-Thomas1973}.
In this context, fixed point configurations correspond to stable patterns of gene expression
at the basis of cellular phenotypes~\cite{B-Thomas1990,J-Wang2012}.
From a theoretical standpoint, their quantity is arguably the most studied dynamical
property of BANs~\cite{J-Aracena2008,J-Aracena2017,dns12,J-Gadouleau2014,J-Remy2008,J-Riis2007,B-Robert1986}.
A famous fact is that fixed points are invariant under any block-sequential update schedule,
which corresponds to partitionning the set of automata into subsets $W_1,\dots,W_p$ (called blocks) and
updating synchronously the automata in $W_1$, then the ones in $W_2$, \emph{etc},
and eventually the automata in the subset $W_p$~\cite{B-Robert1986}.
In this paper we study the block-parallel update schedules, which are dual to the block-sequential,
and have the ability to create new fixed points compared to this well established invariance.

Update schedules received a great attention, and debates are still open regarding the most
realistic schedule. Indeed, the seminal works~\cite{J-Kauffman1969,J-Thomas1973}
considered different update schedules.
A hierarchy based on a formal notion of simulation has been established,
with important members of increasing expressivity behing
the parallel update schedule, the sequential ones (permutations of the automata),
the block-sequential ones, the periodic ones (any sequence of subsets, repeatedly applied)
and the deterministic ones~\cite{J-Pauleve2022}.
Bloc-parallel update schedules are the lowest in this hierarchy with the hability to create fixed points
(just above sequential), at a small cost
(size at least 5 for a simple cycle, according to our experiments).

Instead of blocks of automata applied in parallel within each block
and sequentially the different blocks,
block-parallel update schedules apply the local functions sequentially within each block
and in parallel the different blocks (called ordered-block, to avoid confusion).
The number of substeps, instead of being the number of blocks, becomes the least common multiple
of the size of the ordered-blocks (the first time the same set of automata will be updated again).
They have been introduced and justified in~\cite{J-Demongeot2020},
in relation to the ability of chromatin clocks to influence the speed or rate at which
the expression state of genes evolves within the cell nucleus.
The chromatin organization (proteins in contact with the DNA)
plays an important role in the transcription machinery,
as it may disable the accessibility of DNA strands~\cite{J-Zhang2020}.
This process can enforce new attractors to appear (or change in size), as studied in the present work.

\paragraph{Outline}

Section~\ref{s:def} presents all the necessary definitions, with an example of
block-parallel update schedule creating two new fixed points on the positive cycle of size $5$.
Section~\ref{s:art} reviews previous works, and Section~\ref{s:newfp} contains our contributions.
In Subsection~\ref{ss:arbit} we define a simple family of examples creating new fixed points
on disconnected interaction graphs;
in Subsection~\ref{ss:count} we explain how to efficiently count fixed points on positive cycles
from the interaction graph of its parallelization;
in Subsection~\ref{ss:cycle_num} we expose exhaustive numerical experiments quantifying
the creation of new fixed points on positive cycles up to size $11$;
and in Subsection~\ref{ss:cycle_thm} we provide explicit families of block-parallel update schedules
creating exponentially many new fixed points on positive cycles.
Section~\ref{s:conc} concludes and gives perspectives.

\section{Definitions}
\label{s:def}

Let $\entiers{n}=\{0,\dots,n-1\}$ and $\int{k}=\{1,\dots,k\}$.
The definitions are illustrated with an example at the end of the section.

\paragraph{Boolean automata networks}
A \emph{Boolean automata network} (BAN) of size $n$ is a function $f$
on the configuration space $\bool^n$.
Each Boolean component $i\in\entiers{n}$ is called an \emph{automaton},
and in a configuration $x\in\bool^n$ it is in a \emph{state} denoted $x_i\in\bool$.
The individual behavior of each automaton $i\in\entiers{n}$ is described by a \emph{local function}
$f_i:\bool^n\to\bool$, and a dynamics will be obtained when provided with an order (schedule)
on the applications of these local functions to update the states of the automata.
A BAN is formally defined as the collection of its $n$ local functions $\left(f_i\right)_{i\in\entiers{n}}$,
or equivalently by grouping them into $f:\bool^n\to\bool^n$ where
$f(x)_i = f_i(x)$ for any $x$ and $i$.

\paragraph{Interaction graphs}
The architecture of a BAN $f$ is captured in its \emph{interaction graph} denoted $\inter{f}=(V_f,A_f)$,
representing the effective dependencies among its automata.
We have $V_f=\entiers{n}$ \emph{i.e.}~one vertex per automaton,
and an arc from $i$ to $j$ whenever the local function of $j$ depends on the state of $i$
(the in-neighbors of $j$ are the essential variables~\cite{B-Crama2011} of the Boolean function $f_j:\bool^n\to\bool$).

Cycles in the interaction graph play a central role in automata network theory
(discussed in the state of the art section).
In this paper, we will explore the simplest form of interaction with feedback
and its new capabilities through block-parallel updates:
the canonical \emph{positive cycle} of size $n$.
It is the BAN $\bancyclen$ defined by the local functions $\bancyclen_i(x)=x_{i-1 \mod n}$ for $i\in\entiers{n}$,
with $-1 \mod n = n-1$.
Its interaction graph is the simple directed cycle of length $n$, denoted $\cyclen$
(on vertex set $\entiers{n}$ with arcs $\{ (i-1 \mod n,i) \mid i\in\entiers{n}\}$.

\paragraph{Update schedules}
Local functions provide a static description of some BAN, and its actual dynamics greatly depends on
the order of automata updates. This is provided by an \emph{update schedule},
the two most famous being the parallel (fully synchronous update of all states at each step \cite{J-Kauffman1969})
and asynchronous (non-deterministic update of any, but exactly one,
automaton at each step \cite{J-Thomas1973}) schedules.
More complex schedules mix these two approaches, and in the deterministic landscape the notion of \emph{block}
is central: it consists in a subset of automata $W\subseteq\entiers{n}$ whose states are updated simultaneously
during a \emph{substep}.
We denote $\mmj{f}{W}:\bool^n\to\bool^n$ the function such that:
\[
  \mmj{f}{W}(x)=\begin{cases}
    f_i(x) &\text{if } i\in W,\\
    x_i &\text{otherwise}.
  \end{cases}
\]

\paragraph{Block-sequential update schedules}
A \emph{block-sequential update schedule} is a sequence of blocks updated sequentially
(one block during each substep), such that each automaton is updated exactly once during a step.
It is formally given as an ordered partition $\mu=(W_1,\dots,W_p)$ of $\entiers{n}$,
and its associated \emph{transition function} or \emph{dynamics} on configuration space $\bool^n$
is defined as $\mmjblock{f}{\mu} = \mmj{f}{W_p}\circ\ldots\circ\mmj{f}{W_1}$,
where $\circ$ is the standard composition of functions.
The parallel schedule is a particular case of block-sequential update schedule, namely $\mu=(\entiers{n})$
with a unique block.
Let $\BSn$ denote the set of block-sequential update schedules possible for BANs of size $n$.

\paragraph{Block-parallel update schedules}
\emph{Block-parallel update schedules} are dual of block-sequential schedules:
the updates are sequential within one block, and the blocks are updated synchronously.
We formally employ the term \emph{o-block} to define them, it is a sequence of automaton,
\emph{i.e.}~an element of $\entiers{n}^*$.
A partitioned order $\mu=\{S^1,\dots,S^k\}\subseteq\entiers{n}^*$ is a finite set of o-blocks
(where $S^i_j$ will denote the $j$-th element of the sequence $S^i\in\entiers{n}^*$,
with $j\in\int{|S^i|}$),
such that the o-blocks have no automaton in common,
and the union of their automata is $\entiers{n}$.
The first substep will update the first element of each o-block,
the second substep will update the second element of each o-block,
\emph{etc} (considering the automata within each one o-block cyclically),
until the last elements of all o-blocks are updated during the same substep,
which occurs at the least common multiple of their sizes substep.
Formally, it is convenient to ``sequentialize'' a block-parallel update schedule,
via the morphism $\varphi$ from partitioned orders to sequences of blocks, defined as
$\varphi(\mu) = (W_1,\dots,W_\ell)$ with $\ell = \lcm(|S^1|,\dots,|S^k|)$ and:
\[
  W_i = \{ S^j_{\left(i-1 \mod |S^j|\right)+1} \mid j\in\int{k} \} \text{ for } i\in\int{\ell}.
\]
The \emph{transition function} or \emph{dynamics} is then defined as
$\mmjoblock{f}{\mu} = \mmjblock{f}{\varphi(\mu)} = \mmj{f}{W_\ell}\circ\ldots\circ\mmj{f}{W_1}$.
Let $\BPn$ denote the set of block-parallel udpate schedules possible for BANs of size $n$.
Again, the parallel update schedule is an element of $\BPn$, corresponding to the set of $n$ singleton sequences
$\mu=\{(0),(1),\dots,(n-1)\}$.
The intersection of $\BPn$ and $\BSn$ is characterized in~\cite{C-Perrot2024a}.

\paragraph{Fixed points}
Given a BAN $f$ of size $n$ and an update schedule $\mu$,
we denote the dynamics $\mmjblock{f}{\mu}$ when $\mu\in\BSn$,
and $\mmjoblock{f}{\mu}$ when $\mu\in\BPn$.
In both cases the dynamics is deterministic, and we abstractely denote it $\mmj{f}{\mu}$
when $\mu$ may be block-sequential or block-parallel.
A \emph{fixed-point} is a configuration $x\in\bool^n$ such that $\mmj{f}{\mu}(x)=x$.

\paragraph{Parallelization}
Given a BAN $f$ of size $n$ and an update schedule $\mu$,
the function $\mmj{f}{\mu}:\bool^n\to\bool^n$ itself is called their \emph{parallelization}.
It can be obtained through the process of computing the substeps,
that is for $(W_1,\dots,W_\ell)\in\BSn$ or $\varphi(\mu)=(W_1,\dots,W_\ell)$ where $\mu\in\BPn$,
to consider inductively:
\[
  \mmjsub{f}{\mu}{0}(x) = x
  \quad\text{and}\quad
  \mmjsub{f}{\mu}{i}(x) = \mmjblock{f}{W_i}(\mmjsub{f}{\mu}{i-1}(x))
  \text{ for }
  i\in\int{\ell}.
\]
Properties on the structure of the sequence of interaction graphs:
\[
  (\inter{\mmjsub{f}{\mu}{i}})_{i\in\int{\ell}}
  \text{ with }
  \inter{\mmjsub{f}{\mu}{\ell}}=\inter{\mmj{f}{\mu}}
\]
will be useful to study the dynamics of cycles under block-parallel update schedules.

With block-sequential updating schedules, the value of $\ell$ is capped by $n$
and the parallelization can be computed in polynomial time
(given as inputs $f$, $\mu$ and $x$, output $\mmjblock{f}{\mu}(x)$)~\cite{C-Perrotin2023},
but with block-parallel update schedules the number of substeps may grow exponentially with $n$
and computing the parallelization is hard for $\PSPACE$
(even a single bit of $\mmjoblock{f}{\mu}(x)$)~\cite{C-Perrot2024b}.

\paragraph{Example}

Let $f^{5+}$ the
canonical positive cycle of size $n=5$,
and the two block-parallel updates schedules $\mu_{par} = \{(0), (1), (2), (3), (4)\}$ and
$\mu_{new}= \{(0,1), (2, 3, 4)\}$.
Observe that $\mu_{par}$ is the parallel update schedule, hence
$f^{5+}_{\mu_{par}}$ has two fixed points: $00000$ and $11111$.
The schedule $\mu_{new}$ is more complex:
it is not equivalent to a block-sequential update schedule,
because it has repetitions (\emph{i.e.}, automata update more than once during one step)
in its sequence of blocks form:
$\varphi(\mu_{new}) = (\{0, 2\}, \{1, 3\}, \{0, 4\}, \{1, 2\}, \{0, 3\}, \{1, 4\})$.

For the sake of simplicity, we will denote $f^{5+}_{\mu_{new}}$ as $g$ to expose
the parallelization process.
We start with $g^{(0)}(x) = x$.
\[
  \begin{array}{ccc}
g^{(1)}(x) : \left\lbrace
  \begin{array}{l}
    g^{(1)}_0(x) = x_4\\
    g^{(1)}_1(x) = x_1\\
    g^{(1)}_2(x) = x_1\\
    g^{(1)}_3(x) = x_3\\
    g^{(1)}_4(x) = x_4\\
  \end{array}
\right.
    &
g^{(2)}(x) : \left\lbrace
  \begin{array}{l}
    g^{(2)}_0(x) = x_4\\
    g^{(2)}_1(x) = x_4\\
    g^{(2)}_2(x) = x_1\\
    g^{(2)}_3(x) = x_1\\
    g^{(2)}_4(x) = x_4\\
  \end{array}
\right.
    &
g^{(3)}(x) : \left\lbrace
  \begin{array}{l}
    g^{(3)}_0(x) = x_4\\
    g^{(3)}_1(x) = x_4\\
    g^{(3)}_2(x) = x_1\\
    g^{(3)}_3(x) = x_1\\
    g^{(3)}_4(x) = x_1\\
  \end{array}
\right.
    \\[3.5em]
g^{(4)}(x) : \left\lbrace
  \begin{array}{l}
    g^{(4)}_0(x) = x_4\\
    g^{(4)}_1(x) = x_4\\
    g^{(4)}_2(x) = x_4\\
    g^{(4)}_3(x) = x_1\\
    g^{(4)}_4(x) = x_1\\
  \end{array}
\right.
   &
g^{(5)}(x) : \left\lbrace
  \begin{array}{l}
    g^{(5)}_0(x) = x_1\\
    g^{(5)}_1(x) = x_4\\
    g^{(5)}_2(x) = x_4\\
    g^{(5)}_3(x) = x_4\\
    g^{(5)}_4(x) = x_1\\
  \end{array}
\right.
   &
g^{(6)}(x) : \left\lbrace
  \begin{array}{l}
    g^{(6)}_0(x) = x_1\\
    g^{(6)}_1(x) = x_1\\
    g^{(6)}_2(x) = x_4\\
    g^{(6)}_3(x) = x_4\\
    g^{(6)}_4(x) = x_4\\
  \end{array}
\right.
  \end{array}
\]

This process is illustrated in Figure~\ref{fig:example_cycle5},
starting with the interaction graph of $g^{(0)}(x)$ which is the identity function.
The resulting graph has two
connected components, consequently $f^{5+}_{\mu_{new}}$ has two more fixed
points compared to $f^{5+}_{\mu_{par}}$: $00111$ and $11000$.

\begin{figure}
  \centering
  \includegraphics[scale=0.65]{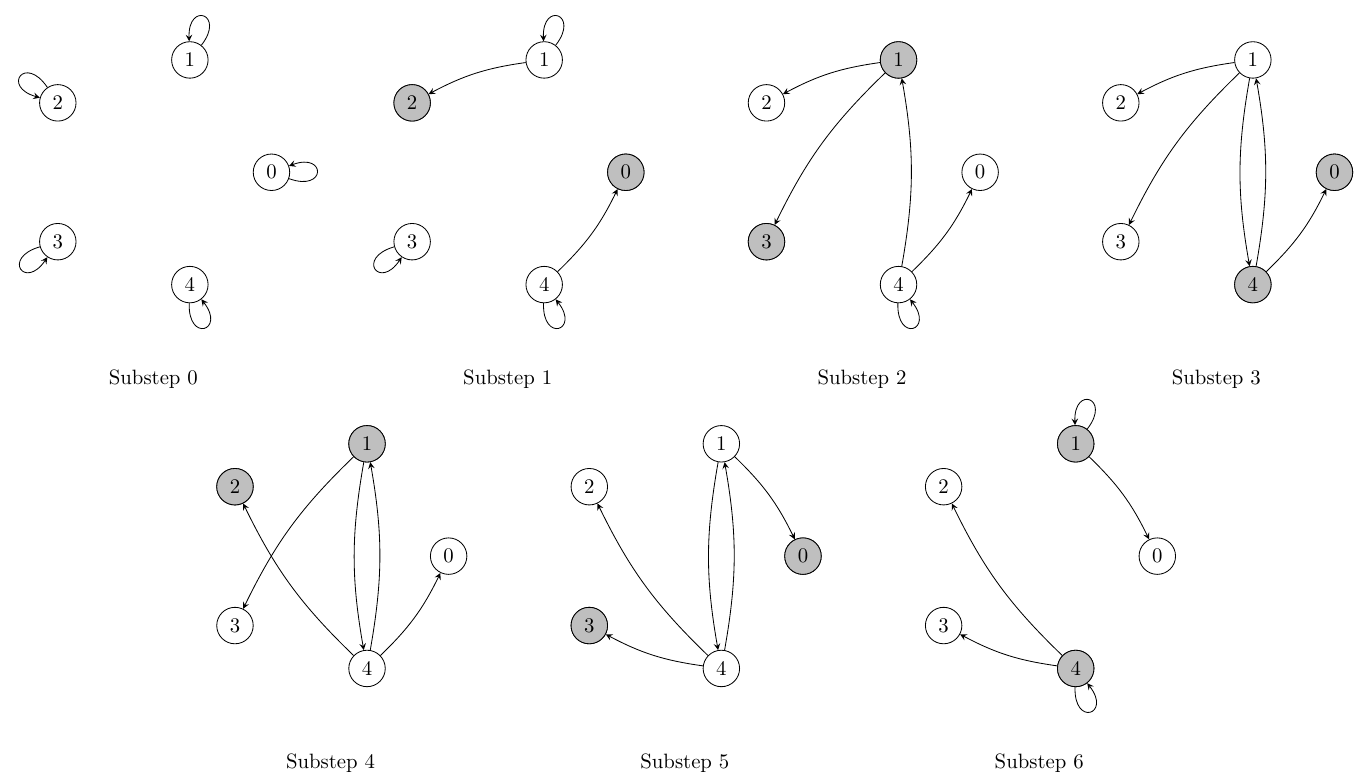}
  \caption{
    Substeps of the parallelization of $f^{5+}$ for the block-parallel update schedule
    $\mu_{new}$. The automata updated at each substep are greyed.
  }
  \label{fig:example_cycle5}
\end{figure}

\section{State of the art}
\label{s:art}

A seminal result of Robert states that cycles in the interaction graph (so called feedbacks)
are the engine of dynamical complexity in Boolean automata networks.
Indeed, a BAN $f$ of size $n$ with $\inter{f}$ acyclic converges towards
a unique fixed point in at most $n$ steps~\cite{J-Robert1980,Gadouleau2023}.
The size of a minimal feedback vertex set in $\inter{f}$ also bounds the maximum number of fixed points
obtained under parallel and asynchronous update schedules~\cite{J-Aracena2008,J-Riis2007}.

One easily notices that if a configuration $x$ is a fixed point of $f$ for the parallel update schedule
(\emph{i.e.}~$f(x)=x$), then $x$ is also a fixed point for any update schedule $\mu$
(\emph{i.e.}~$\mmj{f}{\mu}(x)=x$)~\cite{B-Goles1990}.
Within block-sequential update schedules the fixed point invariance is stronger:
for any $\mu\in\BSn$ the set of fixed points of $\mmjblock{f}{\mu}$ is exactly the same~\cite{B-Robert1986}.
Since the parallel update schedule is an element of $\BSn$, this can be stated as
$f(x)=x$ if and only if $\mmjblock{f}{\mu}(x)=x$.
In the present work we study how this invariance is broken under block-parallel update schedules,
where new fixed points can emerge in the dynamics.

Block-parallel update schedules have been introduced in~\cite{J-Demongeot2020},
where the authors show examples of fixed points creation
(and suggest biological interpretations).
This is the track we will develop in the rest of this article.
The theoretical understanding of block-parallel update schedules has been developed in~\cite{C-Perrot2024a},
providing formulas to count them and algorithms to enumerate them
(exploited for our numerical simulations in Subsection~\ref{ss:cycle_num}).
Update repetitions allow to have an exponential number of substeps within a single iteration
(\emph{cf.}~our discussion on $\varphi$ above), and in~\cite{C-Perrot2024b} it is proven that
the computational complexity of classical problems tend to be lifted to the level of $\PSPACE$,
although this is not a universal rule (unless major complexity collapses).

Let us recall two simple observations.
First, the positive cycle BAN $\bancyclen$ of size $n$ whose interaction graph is $\cyclen$,
has two fixed points in parallel: the two uniform configurations
denoted $0^n$ and $1^n$.
Second, if we modify $\bancyclen$ to set the local function of one $i$ as $f(x)_i=\neg x_{i-1}$
(instead of $\bancyclen_i(x)=x_{i-1}$), we obtain a negative cycle which
has no fixed point in parallel.
An exhaustive study of their limit dynamics is exposed in~\cite{dns12}.

\section{New fixed points for block-parallel update schedules}
\label{s:newfp}

As an introduction to this section exposing our contributions,
we present on Figure~\ref{fig:example_3} a minimal example
of block-parallel update schedule for an automata network of size $3$
which creates two new fixed points compared to any block-sequential update schedule.
The pitfall of this example is that its interaction graph is disconnected,
as the first famility exposed in Subection~\ref{ss:arbit}.
Subsections~\ref{ss:count}, \ref{ss:cycle_num} and~\ref{ss:cycle_thm}
focus on positive cycles, the simplest non-trivial connected structure of interaction
for the study of fixed points.

\begin{figure}
  \centering
  \includegraphics{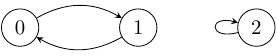}
  \caption{
    Example automata network of size $n=3$ with local functions
    $f_1(x)=x_2$, $f_2(x)=x_1$ and $f_3(x)=\neg x_3$.
    It has no fixed point in parallel (hence for any block-sequential update schedule),
    whereas for the block-parallel update schedule $\mu=\{(0,1),(2)\}$
    we have $\varphi(\mu)=(\{0,2\},\{1,2\})$ which admits the four fixed points
    $000$, $001$, $110$, and $111$.
  }
  \label{fig:example_3}
\end{figure}

\subsection{From zero to an exponential number of fixed points}
\label{ss:arbit}

In this section we provide a simple family of Boolean automata networks and block-parallel update schedules,
creating arbitrarily many new fixed points
(a quantity exponential in the size of the network and in the number of substeps)
compared to any block-sequential update schedule.

For any integer $n>0$, let $\arbitfn$ denote the BAN of size $3n$
whose interaction graph is composed of a negative cycle on $n$ automata,
plus $2n$ isolated nodes.
Formally, its local functions are defined as (interaction graph on Figure~\ref{fig:arbit}):
\begin{align*}
  & \arbitfn_0(x) = \neg x_{n-1},\\
  & \arbitfn_i(x) = x_{i-1} \text{ for all } i\in\entiers{n-1}\setminus\{0\},\\
  & \arbitfn_i(x) = 0 \text{ for all } i\in\entiers{3n-1}\setminus\entiers{n-1}.
\end{align*}
Under the parallel update schedule, or any block-sequential schedule, it has no fixed point
(because of the negative cycle).

\begin{figure}
  \centering
  \includegraphics{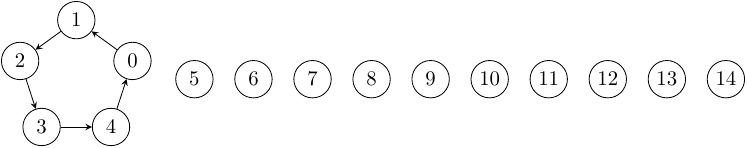}
  \caption{
    Illustration of $\inter{\arbitfn}$, the interaction graph of $\arbitfn$, for $n=5$.
  }
  \label{fig:arbit}
\end{figure}

\begin{theorem}
  For any integer $n>0$, under the block-parallel update schedule:
  \[
    \arbitmun=\{(0),(1),(2),\dots,(n-2),(n-1),(n,n+1,n+2,\dots,3n-2,3n-1)\},
  \]
  the dynamics $\mmjoblock{\arbitfn}{\arbitmun}$ has $2^n$ fixed points.
\end{theorem}

\begin{proof}
  Automata from the set $\entiers{3n-1}\setminus\entiers{n-1}$ have constant $0$ local functions,
  hence their state in any fixed point is $0$.
  The least common multiple of o-block sizes, \emph{i.e.}~the number of substeps, is $\ell=2n$.
  From any initial configuration $x\in\bool^{3n}$,
  after $n$ substeps, each automaton in the cycle has its own initial state negated.
  For simplicity we denote $\arbitfn$ as $\arbitf$ and $\arbitmun$ as $\arbitmu$, thus
  for all $i\in\entiers{n-1}$, we have $\mmjsub{\arbitf}{\arbitmu}{n}(x)_i = \neg x_i$.
  It follows that after $n$ additional substeps, these automata get back to their initial state, hence
  for all $i\in\entiers{n-1}$:
  \[
    \mmjoblock{\arbitf}{\arbitmu}(x)_i = \mmjsub{\arbitf}{\arbitmu}{2n}(x)_i = x_i.
  \]
  We deduce that any $\tilde{x}\in\bool^{3n}$ with $\tilde{x}_i=0$ for all $i\in\entiers{3n-1}\setminus\entiers{n-1}$
  is a fixed point.
\end{proof}

Observe that, redefining $\arbitfn_0(x)=x_{n-1}$ (without negation)
and removing $n$ constant automata, we obtain an BAN of size $2n$
with two fixed points in parallel (for any $n>0$),
and still $2^n$ fixed points under $\arbitmun\in\BPn$,
but now $\arbitmun$ has only $n$ substeps.

\subsection{Counting fixed points on positive cycles}
\label{ss:count}

In this subsection we present two results employed to easily count
the number of fixed points obtained through the parallelization of positive cycles
(in Subsections~\ref{ss:cycle_num} and~\ref{ss:cycle_thm}),
whose interaction graphs have in-degree one.

\begin{theorem}
  \label{theorem:indegree_one}
  Let $f$ be a BAN of size $n$. If $\inter{f}$ has in-degree one,
  then for any $\mu\in\BPn$ the parallelization
  $\inter{\mmjoblock{f}{\mu}}$ also has in-degree one.
\end{theorem}

\begin{proof}
  We prove the result by induction on the number of substeps under $\mu$,
  by noticing that $\inter{f}$ of in-degree one means that every local
  function is of the form $f_i(x)=x_j$ or $f_i(x)=\neg x_j$ for some $j\in\entiers{n}$
  (hence we abusively say that the local function has in-degree one).
  Initially $\mmjsub{f}{\mu}{0}$ is the identity, verifying the claim.
  If $\mmjsub{f}{\mu}{s}$ has in-degree one,
  automata not udpated will still have the same local function of in-degree one,
  whereas the update of automaton $i$ of local function $f_i(x)=x_j$
  will let:
  \[
    \mmjsub{f}{\mu}{s+1}(x)_i=\mmjsub{f}{\mu}{s}(x)_j
  \]
  which has in-degree one by induction hypothesis.
\end{proof}

\begin{theorem}
  \label{theorem:cyclefp}
  Let $f$ be a BAN of size $n$. If $\inter{f}$ has in-degree one
  and only positive arcs, \emph{i.e.}~for all $i\in\entiers{n}$
  the local function is $f_i(x)=x_j$ for some $j\in\entiers{n}$,
  then $f$ has $2^c$ fixed points, with $c$ the number of cycles in $\inter{f}$
  (which equals its number of components).
\end{theorem}

\begin{proof}
  Since $\inter{f}$ has in-degree one, each of its component is a cycle
  with down-trees rooted on it.
  Components are independent, and each of them has $2$ fixed points:
  the all-$0$ and the all-$1$ configurations, because they are the only
  fixed points on the cycle and spread to the trees (from root to leaves).
  The number of combinations is $2^c$.
\end{proof}

\subsection{Numerical simulations on positive cycles}
\label{ss:cycle_num}

\begin{figure}
 \centering
 \begin{tabular}{| c | c | c | c | c | c | c |}
  \hline
  $n$ & 1 cycle & 2 cycles & 3 cycles & 4 cycles & 5 cycles & 6 cycles \\
  \hline
  3 & 13 & 0 & 0 & 0 & 0 & 0\\
  \hline
  4 & 67 & 0 & 0 & 0 & 0 & 0\\
  \hline
  5 & 441 & 30 & 0 & 0 & 0 & 0\\
  \hline
  6 & 3555 & 36 & 0 & 0 & 0 & 0\\
  \hline
  7 & 29625 & 3360 & 588 & 0 & 0 & 0\\
  \hline
  8 & 293091 & 30552 & 5400 & 0 & 0 & 0\\
  \hline
  9 & 3401113 & 424278 & 73296 & 20700 & 0 & 0\\
  \hline
  10 & 42263483 & 4757460 & 629950 & 172900 & 1800 & 1500\\
  \hline
  11 & 551305591 & 83321513 & 20529729 & 7008540 & 1133550 & 130680\\
  \hline
  \end{tabular}
  \caption{
    Results of the numerical simulations on the number of fixed points,
    for block-parallel update schedules on canonical positive cycles.
    For $n$ from $3$ to $11$, the count in column ``$c$ cycles''
    is the number of different schedules whose parallelization has $c$ cycles
    (hence $2^c$ fixed points by Theorem~\ref{theorem:cyclefp}).
  }
  \label{fig:table_numericalsim}
\end{figure}

In this subsection we present numerical simulations to count the number of fixed points
that block-parallel update schedules can reach on the primitive structures of canonical positive cycles.
More precisely, for each size $n$, we study all block-parallel update schedules from $\BPn$
that give different sequences of blocks through $\varphi$
(called representatives for the equivalence classes of $\equiv_0$ in~\cite{C-Perrot2024a}).
For each $n$, we enumerate the schedules and sort them according to the number
of cycles of the parallelized interaction graph.
From Theorem~\ref{theorem:cyclefp}, the number of cycles is the base 2 logarithm
of the number of fixed points.

Algorithmically, the enumerator is taken from~\cite{C-Perrot2024a},
and the parallelization proceeds substep by substep according to Theorem~\ref{theorem:indegree_one}
(updating the in-neighbor of each local function).
For in-degree one parallelized interaction graphs,
the number of components (equal to the number of cycles)
is computed in linear time.

Our code in Python is accessible on the following repository:
\begin{center}
  \url{https://framagit.org/leah.tapin/blockpargenandcount}.
\end{center}
It is archived by Software Heritage at the following permalink:
\begin{center}
  \url{https://archive.softwareheritage.org/browse/origin/directory/?origin_url=https://framagit.org/leah.tapin/blockpargenandcount}.
\end{center}
We have conducted these numerical experiments on a standard laptop
(processor Intel-Core$^\text{TM}$ i7 @ 2.80 GHz),
it took around $3$ hours for $n=11$.

The results are presented on Figure~\ref{fig:table_numericalsim}.
The maximum number of cycles for $n=10$ and $n=11$ is $6$.
One can remark that the increase of the maximum number of cycles ($1,1,2,2,3,3,4,6,6$)
is irregular (we got no appropriate record on \emph{OEIS},
because the number of cycles is at most $n$).
We have observed through further experimentations that the maximum number
of cycles seems to correspond to update schedules with the largest number of substeps.

\subsection{Creation of new fixed points on positive cycles}
\label{ss:cycle_thm}

Recall that $\bancyclen$ is the BAN whose interaction graph
is the cycle of size $n$ with only positive arcs (denoted $\cyclen$).
In parallel, it has two fixed points (the two uniform configurations).
In this section we present two families of update schedules (depending on the parity of $n$),
showing that even on such a primitive architecture of interactions,
it is possible to get new fixed points under block-parallel update schedules.

\begin{theorem}\label{theorem:mmj_oddn}
  Let $n = 2k+1$ for some $k\geq 2$, and:
  \[
    \mu_{odd} = \{(0, 1, \ldots, k-1),
    (k, k+1, 2k, 2k-1, \ldots, k+2)\}.
  \]
  Then the interaction graph of $\mmjoblock{\bancyclen}{\mu_{odd}}$ has $k$ cycles,
  and $2^k$ fixed points.
\end{theorem}

\begin{proof}
\begin{figure}[t]
 \centering
 \includegraphics[scale=0.8]{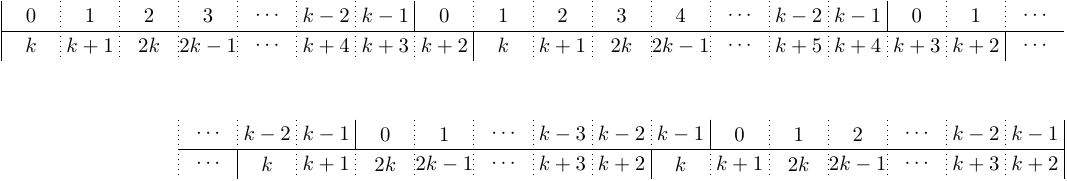}
 \caption{
   Illustration of the ``sequentialization'' of
   $\mu_{odd}$, focusing on its beginning and ending.
   The temporality of substeps flows from left to right,
   each substep updates two automata:
   $\{0,k\}$, then $\{1,k+1\}$, etc.
   Solid vertical lines emphasize the o-blocks.
 }
 \label{fig:mmj_oddn}
\end{figure}
The number of substeps under $\mu_{odd}$ is $k(k+1)$, see Figure~\ref{fig:mmj_oddn}.
For this proof, we will denote the value of automaton $i$ at substep $t$ as
$x_i^t$, with $i\in\entiers{n}$ and $t\in \{0,\dots,k(k+1)\}$
(we have $x^0=x$).
We partition the $n$ automata into three groups:
$A = \entiers[0]{k}$, $B = \{k,k+1\}$, and $C = \{k+2,\ldots,2k\}$.
We will only pay attention to the substeps that are multiples of $k$
(for group $A$) or $k+1$ (for groups $B$ and $C$),
which are the substeps at the end of which the entire o-block of the
corresponding group has been updated.

The automata of groups $A$ and $B$ are in ascending order. This means that after
each automaton in the group has been updated the same number of times in a step
(which is the case at the substeps we consider, including the final substep) all
automata of each of these two groups have the same value.

The automata of group $C$ are in descending order. This means that instead of a
single value being propagated as in the previous groups, the values are shifted
from an iteration of their o-block to the next one, with the introduction of a new value
through automaton $k+2$.
This can be summarized as follows:
for any $a, b, c$ such that $a, a-c \in \{k+2,\ldots,2k\}$ and
$b, b-c \in \{0,\ldots, k\}$, we have $x_a^{(k+1)b} = x_{a-c}^{(k+1)(b-c)}$.

Back to group $A$, we have (recall Figure~\ref{fig:mmj_oddn}):
\begin{align*}
  x_0^{ki} & = x_{2k}^{(k+1)(i-1)} \text{ for all } i \in \entiers[1]{k-1},\\
  x_0^{k^2} & = x_{2k}^{(k+1)(k-2)},\\
  \text{and }  x_0^{k(k+1)} & = x_{2k}^{(k+1)(k-1)}
  = x_{k+2}^{k+1}
  = x_{k+1}^{k+1}
  = x_{k-1}^0.
\end{align*}
We conclude that at the end of a step, every automaton of group A has taken the value of
automaton $k-1$ from the previous step.

Now consider group $B$, we have:
\begin{align*}
  x_k^{(k+1)i} &= x_{k-1}^{k(i-1)} = x_0^{k(i-1)} \text{ for all } i \in \entiers[1]{k},\\
  \text{and in particular, }  x_k^{(k+1)k} &= x_0^{k(k-1)}
  = x_{2k}^{(k+1)(k-2)}
  = x_{k+2}^0 \text{ from what precedes}.
\end{align*}

Regarding group $C$, we have for any $j\in C$:
\begin{align*}
  x_j^{(k+1)k} & = x_{k+2}^{(k+1)(2k+2-j)} & (2 \leq 2k+2-j\leq k)\\
  & = x_{k+1}^{(k+1)(2k+2-j)} & (2 \leq 2k+2-j\leq k)\\
  & = x_0^{k(2k+1-j)} & (1 \leq 2k+1-j\leq k-1)\\
  & = x_{2k}^{(k+1)(2k-j)}& (0 \leq 2k-j\leq k-2)\\
  & = x_j^{0}.
\end{align*}
This means that, for the update schedule $\mu_{odd}$, the local function for every automaton of group $C$
in $\mmjoblock{\bancyclen}{\mu_{odd}}$ is the identity function.

As a conclusion, the interaction graph of the parallelization $\mmjoblock{\bancyclen}{\mu_{odd}}$ is:
\begin{figure}[H]
  \centering
  \includegraphics[scale=0.8]{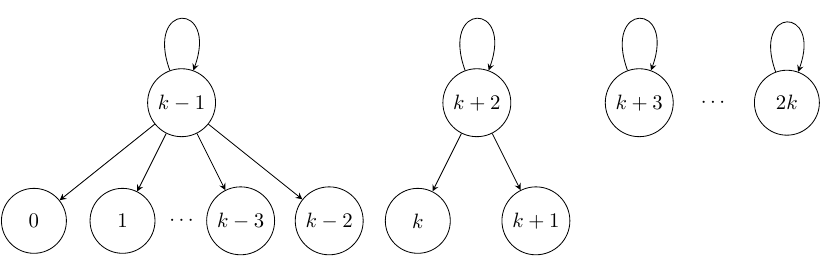}
  \label{fig:intergraphOddnAll}
\end{figure}
\noindent
and this graph contains exactly $k$ cycles,
hence by Theorem~\ref{theorem:cyclefp} it has $2^k$ fixed points.
\end{proof}

\begin{theorem}\label{theorem:mmj_evenn}
  Let $n = 2k$ for some $k\geq 4$, and:
  \[
    \mu_{even} = \{(0), (1, \ldots, k-1),
    (k, 2k-1, k+1, 2k-2, \ldots, k+2)\}.
  \]
  Then the interaction graph of $\mmjoblock{\bancyclen}{\mu_{even}}$ has $k-1$ cycles,
  and $2^{k-1}$ fixed points.
\end{theorem}

\begin{proof}
\begin{figure}[t]
 \centering
 \includegraphics[scale=0.8]{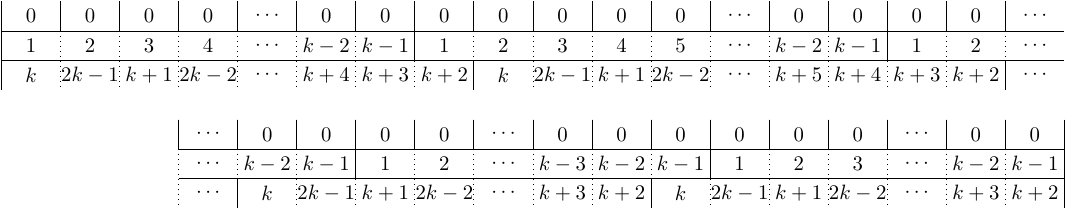}
 \caption{
   The beginning and the end of the ``sequentialization'' of
   $\mu_{even}$
 }
 \label{fig:mmj_evenn}
\end{figure}
The number of substeps under $\mu_{even}$ is $(k-1)(k+1)$,
see Figure~\ref{fig:mmj_evenn}.
We proceed as in the proof of Theorem~\ref{theorem:mmj_oddn},
using the notation $x^t_i$ for the value of automaton $i$ at substep $t$,
with $i\in\entiers{n}$ abd $t\in\{0,\dots,(k-1)(k+1)\}$.
Consider the following partition of the $n$ automata into four groups:
$A = \{0\}$,
$B = \{1, \ldots, k-1\}$,
$C = \{k, k+1\}$, and
$D = \{k+2,\ldots, 2k-1\}$.
We will only pay attention to the substeps that are multiples of $k$
(for group $B$) or $k-1$ (for groups $C$ and $D$)
Group $A$ is a particular case since it contains only automaton $0$, which is also
the only automaton in its o-block, meaning that it is updated at every substep.

Both groups $B$ and $C$ are in ascending order, meaning that they all have the same
value after being updated the same amount of times.
For group $C$, the fact that $k+1$ is not updated right after $k$ is not a problem:
since $k$ doesn't get updated again in the meantime, it keeps the same value.
Group $D$ is in descending order, meaning that we have the same shift of values as
with group $C$ in the proof of Theorem~\ref{theorem:mmj_oddn} for the odd case,
with a slightly different formula:
for any $a, b, c$ such that $a, a-c \in \{k+2,\ldots,2k-1\}$ and
$b, b-c \in \{0,\ldots, k-1\}$, we have $x_a^{kb} = x_{a-c}^{k(b-c)}$.

For group $B$, every $k-1$ subsets, automaton $1$ takes the value of automaton $0$. As we can see
in Figure~\ref{fig:mmj_evenn}, in all but the two last ``cycles'' of subsets, $0$
has taken the value of $2k-1$ after being updated one time less than $1$. In the
two last cycles, it has taken the previous value of $2k-1$ (after being updated
two times less than $1$).
We then have:
\begin{align*}
  x_1^{k-1} & = x_0^0,\\
  x_1^{(k-1)i} & = x_0^{k(i-1)}
  = x_{2k-1}^{k(i-1)} \text{ for all } i\in\{2,\dots,k-2\},\\
  x_1^{(k-1)^2} & = x_{2k-1}^{k(k-3)},\\
  \text{and } x_1^{(k-1)k} & = x_{2k-1}^{k(k-2)}
  = x_{k+2}^k
  = x_{k+1}^k
  = x_k^k
  = x_{k-1}^0.
\end{align*}
We conclude that at the end of a step,
every automaton of group $B$ has taken the value of
automaton $k-1$ from the previous step.

Regarding group $C$, we have:
\begin{align*}
  x_k^{ki} & = x_{k-1}^{(k-1)(i-1)} \text{ for all } i\in\{1,\dots,k-1\},\\
 \text{and in particular, }
  x_k^{k(k-1)} & = x_{k-1}^{(k-1)(k-2)}
  = x_{1}^{(k-1)(k-2)}
  = x_{2k-1}^{k(k-3)}
  = x_{2k-1}^{k(k-3)}
  = x_{k+2}^0.
\end{align*}
That is, at the end of a step, both automata of group $C$ have taken the value of
automaton $k+2$ from the previous step.

Considering group $D$,
we have, for any $j\in D\setminus\{2k-1\}$:
\begin{align*}
  x_j^{k(k-1)} & = x_{k+2}^{k(2k+1-j)} & (3 \leq 2k+1-j\leq k-1)\\
  & = x_{k+1}^{k(2k+1-j)} & (3 \leq 2k+1-j\leq k-1)\\
  & = x_k^{k(2k+1-j)} & (3 \leq 2k+1-j\leq k-1)\\
  & = x_{k-1}^{(k-1)(2k-j)} & (2 \leq 2k-j\leq k-2)\\
  & = x_{1}^{(k-1)(2k-j)} & (2 \leq 2k-j\leq k-2)\\
  & = x_{2k-1}^{k(2k-j-1)} & (1 \leq 2k-j-1\leq k-3)\\
  & = x_j^{0}.
\end{align*}
Except $2k-1$, every automaton of group $D$ regains its original value at the end
of the step.
Since automaton $2k-1$ is not updated at the last substep,
while automaton $0$ is, both have the same
value at the end of a step:
\begin{align*}
  x_0^{k(k-1)} & = x_{2k-1}^{k(k-1)}
  = x_{k+2}^{k\times2}
  = x_{k+1}^{k\times2}
  = x_{k}^{k\times2}
  = x_{k-1}^{k-1}
  = x_{1}^{k-1}
  = x_{0}^{0}.
\end{align*}

As a conclusion, the interaction graph of the parallelization
$\mmjoblock{\bancyclen}{\mu_{even}}$ is:
\begin{figure}[H]
  \centering
  \includegraphics[scale=0.8]{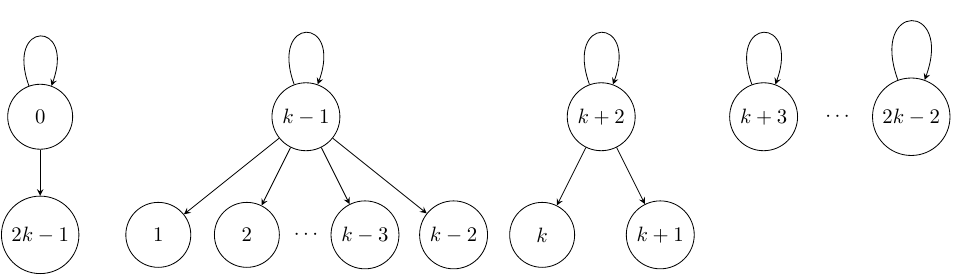}
  \label{fig:intergraphEvennAll}
\end{figure}
\noindent
and this graph contains exactly $k-1$ cycles,
hence by Theorem~\ref{theorem:cyclefp} it has $2^{k-1}$ fixed points.
\end{proof}

\section{Conclusion and further work}
\label{s:conc}

In this article we have studied the creation of fixed points
of Boolean automata networks under block-parallel update schedules:
\begin{itemize}
  \item we have exhibited a simple family of disconnected
    (in terms of the interaction graph)
    BANs that can jump from 0 fixed point in parallel,
    to an exponential number of fixed points in block-parallel;
  \item we have conducted numerical experiments on canonical positive cycles,
    showing that block-parallel schedules are able to create many new fixed points,
    even on these most primitive BANs with feedback;
  \item we have identified families of block-parallel update schedules
    that create exponentially many new fixed points on canonical positive cycles,
    and proved this behavior through the study of the parallelization process
    (and its resulting number of cycles).
\end{itemize}

Theorems~\ref{theorem:mmj_oddn} and~\ref{theorem:mmj_evenn} give exactly
$2^{\lfloor\frac{n-1}{2}\rfloor}$ fixed points (compared to $2$ fixed points in parallel).
However, numerical experiments show that this is not the maximum number of fixed points
attainable under block-parallel update schedules.
Indeed, for $n=10,11$ it is possible to have $2^6$ fixed points.
A perspective would be to characterize the sequence $c(n)$ such that
for $n$ automata it is possible to reach a maximum of $2^{c(n)}$ fixed points
on a positive cycle (from Theorem~\ref{theorem:cyclefp} it is a power of two).
The first terms of $c(n)$ are (starting at $n=3$):
\[
  1,1,2,2,3,3,4,6,6.
\]
There is a trivial bound $c(n)=n$, and obtaining tighter upper bounds would be
an interesting start.

A broader continuation would be to consider more complex interaction graphs of BANs,
and the ability of block-parallel schedules to create fixed points on them.
It suffices to restrict ourselves to connected graphs, because the total number of fixed points
of a BAN is the product of the number of fixed points on each of its components.
In particular, how does the structure of the interaction graph influence the maximum number
of fixed points created by block-parallel update schedules?

\section*{Acknowledgments}

The authors received support from the projects
ANR-24-CE48-7504 ALARICE,
HORIZON-MSCA-2022-SE-01 101131549 ACANCOS, and
STIC AmSud CAMA 22-STIC-02 (Campus France MEAE).

\bibliographystyle{plain}
\bibliography{biblio}

\end{document}